\documentclass{article}
\usepackage{fullpage}
\usepackage{mathtools}
\usepackage{amsmath}
\usepackage{amsfonts}
\usepackage{amsthm}
\usepackage{algorithmic}
\usepackage[ruled]{algorithm2e}
\usepackage{color}
\usepackage{bbm}
\usepackage{enumitem}
\usepackage{xcolor}
\usepackage{hyperref}
\hypersetup{colorlinks,linkcolor={blue},citecolor={blue},urlcolor={red}}  

\newtheorem{theorem}{Theorem}[section]
\newtheorem{lemma}[theorem]{Lemma}
\newtheorem{claim}[theorem]{Claim}
\newtheorem{corollary}[theorem]{Corollary}

\theoremstyle{definition}
\newtheorem{definition}[theorem]{Definition}

\newtheorem{example}[theorem]{Example}

\newcommand{\expected}{\mathbb{E}}

\title{Fine-Grained Privacy Guarantees for Coverage Problems}
\author{Laxman Dhulipala\thanks{University of Maryland, MD, USA. Email: {\tt laxman@umd.edu}. Supported by NSF award number CNS-2317194.} \and George Z. Li\thanks{University of Maryland, MD, USA. Email: {\tt gzli929@gmail.com}. Supported by NSF award number CNS-2317194.}}
\date{}

\begin{document}

\maketitle

\begin{abstract}
    We introduce a new notion of neighboring databases for coverage problems such as Max Cover and Set Cover under differential privacy. In contrast to the standard privacy notion for these problems, which is analogous to node-privacy in graphs, our new definition gives a more fine-grained privacy guarantee, which is analogous to edge-privacy. We illustrate several scenarios of Set Cover and Max Cover where our privacy notion is desired one for the application.
    
    Our main result is an $\epsilon$-edge differentially private algorithm for Max Cover which obtains an $(1-1/e-\eta,\tilde{O}(k/\epsilon))$-approximation with high probability. Furthermore, we show that this result is nearly tight: we give a lower bound show that an additive error of $\Omega(k/\epsilon)$ is necessary under edge-differential privacy. Via group privacy properties, this implies a new algorithm for $\epsilon$-node differentially private Max Cover which obtains an $(1-1/e-\eta,\tilde{O}(fk/\epsilon))$-approximation, where $f$ is the maximum degree of an element in the set system. When $f\ll k$, this improves over the best known algorithm for Max Cover under pure (node) differential privacy, which obtains an $(1-1/e,\tilde{O}(k^2/\epsilon))$-approximation.

    We show that our techniques also apply to the Set Cover problem under differential privacy. Like previous work~\cite{GuptaLMRT10}, we consider the problem of outputting an implicit set cover solution and give a $O(\text{poly}\log{n}/\epsilon)$-approximation algorithm for the problem under $\epsilon$-edge differential privacy. Again via group privacy guarantees, this implies an $O(f\cdot\text{poly}\log{n}/\epsilon)$-approximation algorithm for the problem under $\epsilon$-node differential privacy. In particular, this is the first non-trivial approximation algorithm for outputting an implicit set cover under pure node-differential privacy. Previous work only guarantees the weaker notion of approximate node differential privacy.  
\end{abstract}

\newpage 

\section{Introduction}

Our work considers a fundamental problem in combinatorial optimization called the Max Cover problem. In this problem, we are given a universe $\mathcal{U}=\{u_1,\ldots,u_n\}$ of elements, a set system $\mathcal{S}=\{S_1,\ldots,S_m\}$, and a natural number $k\in\mathbb{N}$. We wish to choose $k$ sets $S_{i_1},\ldots,S_{i_k}$ such that the number of elements covered $|\bigcup_{j=1}^{k}S_{i_j}|$ is maximized. In many applications where the Max Cover problem appears, the input data can be sensitive information which individuals wish to keep private; as a result, we consider the Max Cover problem under the privacy model of differential privacy \cite{Dwork06}, which has become the golden standard for formal data privacy guarantees. For example, consider the following scenario:

\begin{example}
    Suppose that we are living in a global pandemic and we want to increase the accessibility of vaccines to reduce the impacts of the disease. The local government has $k$ vaccine distribution centers, and they want to know where to place them in order to best reduce infections. Suppose they have some data on which popular locations (such as stores, park, or schools) each person visits in the week. One natural objective is to choose $k$ locations such that the number of people which go to at least one of the chosen locations is maximized. This problem can be formulated as an instance of the Max Cover problem, where the universe $\mathcal{U}$ is the set of people and the set system $\mathcal{S}$ consists of subsets of people which visit each location.
\end{example}

Since some people are already vaccinated, our objective should really be to cover the maximum number of people which have not been vaccinated already. However, the individuals may wish to keep this information private (e.g., an individual may not want their friends to know that they have not been vaccinated yet). As a result, we wish to preserve privacy for whether or not an element requires coverage. This notion of privacy, which we call node-differential privacy, has been the primary model in past work \cite{GuptaLMRT10,mitrovic2017differentially}.

\begin{example}
    Given the outbreak of some infectious disease, one commonly used method to detect it is to place sensors at the entrances of popular locations (again, such as stores, parks, or schools) so that infected individuals can be detected \cite{eubank2004structural}. Since placing and monitoring these sensors can be expensive, let us suppose that the local government can afford to place $k$ sensors in the area. Given the information of the locations each person visits during the week (e.g., using past data), we wish to choose $k$ popular locations such that the maximum number of individuals are checked by some sensor set. This problem can again be formulated as an instance of the Max Cover problem, where the universe $\mathcal{U}$ is the set of individuals and set system $\mathcal{S}$ consists of the subset of people which visit each location.
\end{example}

In this example, we are in the beginning of an outbreak so none of the individuals are vaccinated and all individual require coverage. However, the individuals may wish the keep the locations which they visit private. In the language of differential privacy, they wish to guarantee that an adversary cannot detect whether or not they visited a specific location. In analogy to privacy notions on graphs, we call this notion of privacy for Max Cover edge-differential privacy. We remark that this notion of privacy makes sense in Example 1.1 as well in the setting no individuals have been vaccinated yet. Furthermore, these examples are also naturally motivated for the Set Cover problem, where we no longer have an input $k$ and instead, wish to choose to fewest number of sets so that all elements in $\mathcal{U}$ are covered.

\subsection{Our Contributions}

Our work formalizes the new notion of edge-differential privacy for the Max Cover problem (see Section \ref{sec:prelim}). As we have illustrated in the examples above, there are many applications of the Max Cover problem where our new privacy notion is the natural one.

Our main result is a novel algorithm for the Max Cover problem under $\epsilon$-edge differential privacy. The algorithm's approximation guarantee nearly match the multiplicative approximation ratios of non-private algorithms for Max Cover; we also show that the additive error of the algorithm is near optimal. Our algorithm is inspired by recent work on transforming parallel algorithms into differentially private algorithms for problems including densest subgraph and k-core decomposition \cite{DBLP:conf/focs/DhulipalaLRSSY22}. The authors observed that parallel algorithms are naturally amenable to private implementations because they have low round complexity and often use only local information. Motivated by this, we study parallel implementations of the classical greedy Max Cover algorithm and show that they can be made differentially private. Our results here corroborate the observations made in \cite{DBLP:conf/focs/DhulipalaLRSSY22}, contributing a new example of a parallel algorithm which achieves near-optimal additive error under differential privacy.

\begin{theorem}[informal]
For any $\eta>0$, there exists an $(1-1/e-\eta,\tilde{O}(k/\epsilon))$-approximation algorithm for the Max Cover problem under $\epsilon$-edge differential privacy. \label{thm:intro-max-cover}
\end{theorem}

\begin{theorem}
    Any algorithm for the Max Cover problem which satisfies $\epsilon$-edge differential privacy must incur an additive error of $\Omega(k\log(n/k)/\epsilon)$.
\end{theorem}

We also consider the related Set Cover problem. In this problem, we also have a universe of elements $\mathcal{U}$ and a set system $\mathcal{S}$. The goal is to choose the fewest number of sets in $\mathcal{S}$ such that the union of the sets covers all of the elements\footnote{Formally, our objective in Set Cover is instead to choose the fewest number of sets in $\mathcal{S}$ to cover all elements in $\bigcup_{S\in\mathcal{S}}S$. In the non-private case, we can reduce to the case where $\mathcal{U}=\bigcup_{S\in\mathcal{S}}S$ without loss of generality since we can detect when an element is not covered by any set. Since this is no longer possible with privacy, our goal will be to cover all elements is contained in a set $S\in\mathcal{S}$. When analyzing implicit set covers, this is equivalent to assuming that $\mathcal{U}=\bigcup_{S\in\mathcal{S}}S$.}. We show that our techniques for Max Cover also help in developing edge-differentially private algorithms for Set Cover. Like before, these are based on a parallel implementation of the standard sequential greedy algorithm for the problem. Since a logarithmic approximation is necessary even in the non-private case, the approximation guarantees for the problem are again nearly-optimal.

\begin{theorem}[informal]
    There exists an $O(\text{poly}\log{n}/\epsilon)$-approximation algorithm for the Set Cover problem under $\epsilon$-edge differential privacy.\label{thm:intro-set-cover}
\end{theorem}

By investigating the relationship between edge-differential privacy and node-differential privacy for coverage problems, we show that our results implies new algorithms satisfying node-differential privacy as well. The guarantees under node-differential privacy depend on a parameter $f$ of the set system, which denotes the maximum degree of an element (i.e., the maximum number of sets in the set system which an element is contained in). We summarize the results below.

\begin{corollary}[informal]
    For any $\eta>0$, there exists an $(1-1/e-\eta,\tilde{O}(fk/\epsilon))$-approximation algorithm for the Max Cover problem under $\epsilon$-node differential privacy.\label{cor:node-privacy-guarantees-max-cover}
\end{corollary}

\begin{corollary}[informal]
    There exists an $O(f\text{poly}\log(n)/\epsilon)$-approximation algorithm for the Set Cover problem under $\epsilon$-node differential privacy.\label{cor:node-privacy-guarantees-set-cover}
\end{corollary}

For comparison, the best known $\epsilon$-node differentially private algorithms for Max Cover incurs an additive error of $\tilde{O}(k^2/\epsilon)$~\cite{GuptaLMRT10}. In the regime where $f\ll k$, our algorithms give the best known additive error for Max Cover under pure node-differential privacy\footnote{We remark that better additive guarantees are known under the weaker notion of approximate node-differential privacy.}. Furthermore, our results for Set Cover are the only known approximation guarantees for the problem under pure node-differential privacy; previous work only guarantee approximate differential privacy.

\subsection{Related Work}

The max cover problem was first considered in the differential privacy model in~\cite{GuptaLMRT10}, where they introduced the node-privacy model for covering problems. In that work, the authors gave a near-optimal algorithm for the problem under approximate differential privacy and also stated the simple $(1-1/e,\tilde{O}(k^2/\epsilon))$-approximation algorithm under pure differential privacy. Since then, there has been a large body of work focusing on the more general problem of submodular maximization in various models of computation~\cite{mitrovic2017differentially,DBLP:conf/icml/RafieyY20,DBLP:conf/aaai/ChaturvediNZ21,DBLP:conf/aistats/Perez-SalazarC21,DBLP:conf/icml/ChaturvediNN23}. When applying each of these results to the special case of max cover, all of them consider node-differential privacy. 

Differentially private set cover was also first considered in~\cite{GuptaLMRT10}, also under the node-privacy model. Since no guarantees were possible under differential privacy, they introduced the notion of outputting implicit solutions to the problem and gave $O(\text{poly}\log(n)/\epsilon)$-approximation algorithms under this model. Since then, there have been a few attempts at outputting explicit private solutions for set cover via pseudo-approximation. For instance, \cite{DBLP:conf/icalp/HsuRRU14} gave a private $O(\log{n})$-approximation algorithm for explicit set cover, but the solution may leave up to $O(\text{OPT}^2\text{poly}\log(n)/\epsilon)$ elements uncovered. Recently, \cite{DBLP:conf/ijcai/Li0V23} gave a private $O(\text{poly}\log(n)/\epsilon)$-approximation algorithm for explicit set cover, but the solution is a $\rho$-partial set cover. Again, all of these results are under node-differential privacy.

\section{Preliminaries} \label{sec:prelim}

\subsection{Differential Privacy Background}

We first give the definition of differential privacy and describe some basic differentially private mechanisms which we will use as subroutines in our work. We defer readers to \cite{dwork2014algorithmic} for the proofs of these basic results and the motivation and properties of differential privacy. 

Informally, an algorithm is differentially private if the output of the algorithm doesn't differ too much when one individual's data is added or removed from the dataset. This intuitively guarantees that an adversary can't obtain any information about a single individual's private data using the output of a differentially private algorithm. Formally, we say that two datasets are neighboring if they differ by one individual's data, and define differential privacy as follows.

\begin{definition}
    A mechanism $M:\mathcal{X}\to\mathcal{Y}$ is said to be $(\epsilon,\delta)$-differentially private if for any two neighboring inputs $X_1,X_2\in\mathcal{X}$ and any measurable subset of the output space $S\subseteq\mathcal{Y}$, we have that $$\Pr[M(X_1)\in S]\le \exp(\epsilon)\cdot\Pr[M(X_2)\in S]+\delta.$$
    When $\delta=0$, we say that $M$ is $\epsilon$-differentially private.
\end{definition}
 
In our work, we will focus on guaranteeing the stronger notion of $\epsilon$-differential privacy (also called pure differential privacy). We include the definition of $(\epsilon,\delta)$-differential privacy (also called approximate differential privacy) for the sake of discussion of previous work. Also note that the definition of differential privacy requires some suitable notion of neighboring inputs, which models the addition or removal of a single individual's data. We will define this for coverage problems in the next subsection, and this is where our main non-technical contribution lies.

Now, we state the basic properties and subroutines which we will need in the paper. The composition property says that if two mechanisms are private, running both of the mechanisms is still private (with a slightly weaker privacy parameter). The post-processing property says that any post-processing of the output of a private mechanism remains private. Finally, the group-privacy property says that any $\epsilon$-differentially private algorithm still gives some (weaker) privacy guarantees for groups of people.

\begin{lemma}
Let $\mathcal{M}_1:\mathcal{X}\to\mathcal{Y}_1$ and $\mathcal{M}_2:\mathcal{X}\to\mathcal{Y}_2$ be $\epsilon_1$- and $\epsilon_2$-differentially private mechanisms, respectively. Then $\mathcal{M}:\mathcal{X}\to\mathcal{Y}_1\times\mathcal{Y}_2$ defined by $\mathcal{M}=(\mathcal{M}_1,\mathcal{M}_2)$ is $(\epsilon_1+\epsilon_2)$-differentially private.
\label{lem:basic composition}
\end{lemma}

\begin{lemma}
Let $\mathcal{M}:\mathcal{X}\to\mathcal{Y}$ be an $\epsilon$-edge differentially private mechanism.
Let $f:\mathcal{Y}\to\mathcal{Z}$ be an arbitrary randomized mapping. Then $f\circ\mathcal{M}:\mathcal{X}\to\mathcal{Z}$ is still $\epsilon$-edge differentially private.\label{lem:post}
\end{lemma}

\begin{lemma}
    Let $\mathcal{M}:\mathcal{X}\to\mathcal{Y}$ be an $(\epsilon,\delta)$-differentially private mechanism. If $X_1,X_2\in\mathcal{X}$ differ by at most $k$ individual's data (i.e., they are $k$-step neighbors), then we have for all measurable $Y\subseteq\mathcal{Y}$ that
    $$\Pr[\mathcal{M}(X_1)\in Y]\le\exp(k\epsilon)\cdot\Pr[\mathcal{M}(X_2)\in Y]+ke^{(k-1)\epsilon}\delta.$$\label{lem:group}
\end{lemma}

Next, we state the Laplace mechanism, which gives a simple way to approximately output the value of statistic $f:\mathcal{X}^n\to\mathbb{R}^n$ about the data while guaranteeing privacy. All of the algorithm in this paper will be analyzed as compositions of Laplace mechanisms, in addition to some post-processing.

\begin{lemma}[Laplace Mechanism]
    For a vector-valued query $f:\mathcal{X}^n\to\mathbb{R}^n$, define its $\ell_1$-sensitivity to be $\Delta_f=\max_{X_1\sim X_2}\|f(X_1)-f(X_2)\|_1$, where $X_1\sim X_2$ are neighboring inputs. For a given query $f$ and input $X\in\mathcal{X}^n$, define the Laplace mechanism to output
    $$f(X)+(Y_1,\ldots,Y_n),$$
    where $Y_1,\ldots,Y_n$ are independent and $Y_i\sim\text{Lap}(\Delta_f/\epsilon)$. This mechanism is $\epsilon$-differentially private. \label{lem:laplace mechanism}
\end{lemma}

\subsection{Edge-privacy vs. Node-privacy in Coverage Problems}

In this section, we define the new notion of differential privacy for coverage problems, where the private data is a universe $\mathcal{U}$ along with a set system $\mathcal{S}$. 

\begin{definition}
    Given a universe $\mathcal{U}$, two private inputs $\mathcal{S}_1$ and $\mathcal{S}_2$ are said to be \emph{edge-neighbors} if the set systems $\mathcal{S}_1$ and $\mathcal{S}_2$ differ by exactly one set, and that set differs by exactly one element.
\end{definition}

\begin{definition}
    Given a universe $\mathcal{U}$, let $\mathcal{X}$ denote the family of all possible set system and let $\mathcal{Y}$ be any output space.
    A mechanism $\mathcal{M}:\mathcal{X}\to\mathcal{Y}$ is said to be ($\epsilon,\delta$)-\emph{edge differentially private} if for all pairs of edge-neighbors $\mathcal{S}_1, \mathcal{S}_2$ and all measurable subsets $Y\subseteq\mathcal{Y}$, we have
    $$\Pr[\mathcal{M}(\mathcal{S}_1)\in Y]\le\exp(\epsilon)\cdot \Pr[\mathcal{M}(\mathcal{S}_2)\in Y].$$
    If $\delta=0$, we also say that $\mathcal{M}$ is $\epsilon$-\emph{edge differentially private}.
\end{definition}

To motivate the name, observe that we can we view a set system as a bipartite graph. We have a node on the left for each element in the universe $\mathcal{U}$ and a node on the right for each set in the set system $\mathcal{S}$. We then add an edge between a set $S\in\mathcal{S}$ and an element $u\in\mathcal{U}$ if $u\in S$. If we view the input as this bipartite graph representation, then the notion of edge-differential privacy for set systems coincides with standard notion of edge-differential privacy in graphs.

Next, we define the standard notion of neighboring and differentially private set systems used in previous work (see e.g.,~\cite{GuptaLMRT10}). In this notion, we have a public universe and set system $(\mathcal{U},\mathcal{S})$ along with a private subset $R\subseteq\mathcal{U}$ and our goal is to cover all of $R$ while maintaining privacy for which elements are contained in $R$. Since this is analogous to node-privacy when viewing the input as a bipartite graph, we call this node-privacy as well. Formally, it is defined as follows:

\begin{definition}
    Given a set system $(\mathcal{U},\mathcal{S})$, two private inputs $R_1$ and $R_2$ are said to be node-neighboring if they differ by exactly one element (i.e., $|R_1\Delta R_2|=1$).
\end{definition}

\begin{definition}
    Given any set system $(\mathcal{U},\mathcal{S})$, let $\mathcal{X}$ denote the family of all possible inputs $R$ and let $\mathcal{Y}$ be any output space. A mechanism $\mathcal{M}:\mathcal{X}\to\mathcal{Y}$ is said to be $(\epsilon,\delta)$-\emph{node differentially private} if for all pairs of node-neighbors $R_1, R_2$ and all measurable subsets $Y\subseteq\mathcal{Y}$, we have
    $$\Pr[\mathcal{M}(R_1)\in Y]\le\exp(\epsilon)\cdot\Pr[\mathcal{M}(R_2)\in Y].$$
    If $\delta=0$, we also say that $\mathcal{M}$ is $\epsilon$-\emph{node differentially private}.
\end{definition}

Now, we will show the relationship between edge- and node-privacy for covering problems. For ease of exposition, we will state and prove the result for the Max Cover problem. The proof for the Set Cover problem is exactly the same, except that we output an ordering over the subsets, instead of $k$ of the subsets.

\begin{lemma}
    Given an $\epsilon$-edge differentially private mechanism $\mathcal{M}'$ for Max Cover, there is an $f\epsilon$-node differentially private mechanism $\mathcal{M}$ for Max Cover with the same utility guarantees as $\mathcal{M}'$.\label{lem:equivalent-f}
\end{lemma}
\begin{proof}
    Let $(\mathcal{U},\mathcal{S},R)$ be an arbitrary input to the node-private mechanism $M$, and define the Max Cover instance $\mathcal{U}'=\mathcal{U}\cap R$ and $\mathcal{S}'=\{S\cap R:S\in\mathcal{S}\}$ as input for the edge-private algorithm $\mathcal{M}'$. The algorithm $\mathcal{M}$ constructs the Max Cover instance $(\mathcal{U}',\mathcal{S}')$ and applies $\mathcal{M}'$ on it, outputting the $k$ sets chosen by $\mathcal{M}'$. Clearly, $\mathcal{M}$ and $\mathcal{M}'$ have the outputs on any instance of Max Cover, so they give the same multiplicative and additive utility guarantees.

    Next, we turn to the privacy guarantees. Observe that adding or removing one element in $R$ changes at most $f$ sets in $\mathcal{S}'$, and each such set in $\mathcal{S}'$ changes by exactly one element. Indeed, this follows immediately by definition of the maximum degree of the set system $f$. Hence, one node-neighboring instances of $(\mathcal{U},\mathcal{S},R)$, the constructed instances $(\mathcal{U}',\mathcal{S}')$ are $f$-step neighbors. By group privacy guarantees (Lemma \ref{lem:group}), we have that the resulting mechanism $\mathcal{M}$ is $f\epsilon$-differentially private.
\end{proof}

Observe that this implies Corollaries \ref{cor:node-privacy-guarantees-max-cover} and \ref{cor:node-privacy-guarantees-set-cover} assuming the proofs of Theorems \ref{thm:intro-max-cover} and \ref{thm:intro-set-cover}. Specifically, we can compute $f$ for the input set system (which doesn't violate privacy) and run the edge-private algorithms with privacy parameter $\epsilon/f$. The remainder of the paper will focus on giving our edge-private algorithms.

\subsection{Implicit Set Cover Solutions}

It was observed in~\cite{GuptaLMRT10} that any explicit set cover solution which is output by a differentially private algorithm has trivial approximation guarantees. They showed this hardness result for node-differential privacy, and we illustrate below that the same example gives hardness for edge-differential privacy as well.

\begin{example}
    Let $M$ be an $\epsilon$-edge differentially private algorithm for set cover.
    Consider the class of all vertex cover instances where the inputs are graphs $G=(V,E)$. We claim that for any pair of vertices $(u,v)$, the mechanism $M$ must output either $u$ or $v$ (or both) as part of its output. If neither $u$ nor $v$ are part of the output, then an adversary can tell that the edge $(u,v)$ doesn't exist in the graph (since the output is guaranteed to be a vertex cover with high probability). This gives a contradiction since the information on whether $(u,v)\in E$ is private, due to group-privacy guarantees of differential privacy. Thus, we have shown that for each pair $(u,v)$, at least one of the vertices must be contained in the output of an edge-differentially private mechanism. This implies that any private mechanism must output $n-1$ nodes as the vertex cover, which is an essentially trivial algorithm.
\end{example}

Given this hardness result,~\cite{GuptaLMRT10} introduce the notion of an implicit set cover solution where the private algorithm outputs an ordering of the sets. Given the ordering, each element can identify the first set in the ordering which covers them and that set is (implicitly) included in the set cover. The goal is to output an ordering which minimizes the size of this implicit set cover, while satisfying differential privacy. This notion of outputting implicit set covers can more generally be seen as satisfying the \emph{billboard model} or \emph{joint model} of differential privacy (see~\cite{hsu-matching}). Our algorithms for edge-private set cover will also output these implicit representations of the solution.

\section{Maximal Nearly Independent Set}

The main technique introduced by \cite{BPT11} is the notion of a maximal nearly independent set (MaNIS); they show that a MaNIS can be computed in linear work and polylogarithmic depth, and then use the algorithm as a subroutine to give efficient parallel algorithms for several problems. Due to the additive error introduced to guarantee differential privacy, we again need to relax the definition to account for this error. To motivate the definition of a MaNIS, let us recall our goal: we want to choose sets which cover as many additional elements as possible, while guaranteeing parallelism. If we choose any set with utility within a $1+\eta$ factor of optimal, we will have a good approximation guarantee but may not have sufficient parallelism. If we choose all sets with utility within a $1+\eta$ factor of optimal, we have sufficient parallelism but our approximation guarantee can be arbitrarily bad. To satisfy both requirements, we wish to select a maximal collection of sets with good utility (specifically, a $(1+\eta)$-approximation) with a small overlap between sets (i.e., the sets are nearly independent) so that these sets can be chosen in parallel. 

\begin{definition}
Let $\eta>0$ be a constant and $\phi,\psi$ be functions of $n$. Given a bipartite graph $G=(A\cup B,E)$, we say that $J=\{s_1,\ldots,s_k\}\subseteq A$ is a $(\phi,\psi)$-approximate $\eta$-maximal nearly independent set ($\eta$-MaNIS) if
\begin{enumerate}
    \item Nearly Independent: for each index $i\in[k]$, we have $$|N(s_i)\backslash N(\{s_1,\ldots,s_{i-1}\})|\ge (1-4\eta)\cdot|N(s_i)|-\phi.$$
    \item Maximal: for all elements $a\in A\backslash J$, we have $$|N(a)\backslash N(J)|<(1-\eta)\cdot |N(a)|+\psi.$$
\end{enumerate}\label{def:MaNIS}
\end{definition}
\noindent
Now, we will give a differentially private algorithm which outputs an approximate MaNIS.

\begin{algorithm}[h]
\caption{Differentially Private MaNIS}
\label{alg:manis}
\textbf{Input:} Universe $\mathcal{U}$, set system $\mathcal{S}$, privacy parameter $\epsilon$, and multiplicative error $\eta$ \\
\textbf{Output:} $(O(\log^2{n}/\epsilon),O(\log^2{n}/\epsilon))$-approximate $\eta$-MaNIS with high probability \\
\begin{algorithmic}[1]
\STATE Initialize $G^{(0)}=(A^{(0)}\cup B^{(0)},E^{(0)})=(A\cup B,E)$
\STATE Set $\epsilon_0\leftarrow \epsilon/C_1\log(n)$.
\STATE For each $a\in A$, estimate the degree $\tilde{D}(a)=|N_{G^{(0)}}(a)|+\text{Lap}(2/\epsilon_0)$ via the Laplace mechanism.
\FOR{$t=0,\ldots,C_2\log{n}$}
\STATE For $a\in A^{(t)}$, randomly pick $x_a\in[0,1]$
\STATE For $b\in B$, let $\varphi^{(t)}(b)$ be $b$'s neighbor with maximum $x_a$.
\STATE Pick vertices $J^{(t)}=\{a\in A^{(t)}:\sum_{b\in B^{(t)}}\mathbbm{1}\{\varphi^{(t)}(b)=a\}+\text{Lap}(1/\epsilon_0)\ge (1-4\eta)\cdot\tilde{D}(a)\}$
\STATE 
\STATE $B^{(t+1)}=B^{(t)}\backslash N_{G^{(t)}}(J^{(t)})$
\STATE $A^{(t+1)}=\{a\in A^{(t)}\backslash J^{(t)}:|N_{G^{(t)}}(a)\cap B^{(t+1)}|+\text{Lap}(1/\epsilon_0)\ge (1-\eta)\cdot\tilde{D}(a)+\frac{6\log{n}}{(1-3\eta)\epsilon_0}\}$
\STATE $E^{(t+1)}=E^{(t)}\cap (A^{(t+1)}\times B^{(t+1)})$
\STATE Sort $J^{(t)}$ by their $x_a$ value.
\ENDFOR
\STATE Output the concatenation of $J^{(0)},\ldots,J^{(C_2\log{n})}$.
\end{algorithmic}
\end{algorithm}

In order to show that Algorithm \ref{alg:manis} outputs a MaNIS, we first need to show that after the for-loop in Lines 4--12 ends, we have $A^{(t)}=\emptyset$. This will imply that our algorithm indeed outputs an approximate MaNIS (see Theorem \ref{thm:approximate-manis}). Our proofs in this section are based on that of Theorem 3.4 in \cite{BPT11}, but need to be adapted due to the noise added for differential privacy. We begin with the following lemma.
\begin{lemma}
Let $\mathcal{E}$ denote the event that $|\text{Lap}(\beta)|\le 2\beta\log{n}$ for each $\text{Lap}(\beta)$ random variable, and condition on event $\mathcal{E}$.
Define the potential function $\Phi(t)=\sum_{a\in A^{(t)}}|N_{G^{(t)}}(a)|$ which counts the number of remaining edges incident on $A^{(t)}$. Then we have $\expected[\Phi(t+1)|\mathcal{E}]\le (1-c)\cdot\expected[\Phi(t)|\mathcal{E}]$ for $c=\eta^2(1-\eta)/4$.\label{lem:potential function}
\end{lemma}
\begin{proof}
For the entirety of the proof, we will keep the conditioning on $\mathcal{E}$ implicit in the notation for simplicity.
Fix an iteration $t\in\{0,\ldots,C_2\log{n}\}$. Let $\text{deg}(x)=\text{deg}_{G^{(t)}}(x)$ and define $$\Delta_a=\mathbbm{1}\{a\in J^{(t)}\}\textstyle\sum_{b:\varphi^{(t)}(b)=a}\text{deg}(b)$$ to be the sum of degrees of all neighbors of $a$ that are assigned to $a$ by $\varphi^{(t)}(\cdot)$ if $a\in J^{(t)}$, and $0$ otherwise.
Observe that $\sum_{a\in A^{(t)}}\Delta_a$ is a lower bound on the number of edges which are removed during this iteration:
$$\Phi(t+1)-\Phi(t)\ge \textstyle\sum_{a\in A^{(t)}}\Delta_a.$$
As a result, it suffices to show that $\mathbb{E}[\Delta_a]\ge c\cdot\text{deg}(a)$ for each $a\in A^{(t)}$.

Fix $a\in A^{(t)}$ and label the neighbors of $a$ $N_{G^{(t)}}(a)=\{b_1,\ldots,b_{n^\prime}\}$ such that $\text{deg}(b_1)\le\ldots\le\text{deg}(b_{n^\prime})$. We will split our proof into two cases depending on the size of the neighborhood of $a$ in $G^{(t)}$. First, let's consider the case where $n^\prime<2/\eta$. Let $\mathcal{E}_1$ denote the event where $x_a$ is the largest of all $x_{a^\prime}$ for $a^\prime\in N_{G^{(t)}}(N_{G^{(t)}}(a))$. Since $a$ has $n^\prime$ neighbors and each neighbor has at most $\text{deg}(b_{n^\prime})$ neighbors, this implies that $|N_{G^{(t)}}(N_{G^{(t)}}(a))|\le n^\prime\cdot \text{deg}(b_{n^\prime})$.
As a result, we have 
\begin{align}
    \Pr[\mathcal{E}_1]\ge{1}/[{n^\prime\cdot\text{deg}(b_{n^\prime})}].\label{eq:deg-bound}
\end{align}
Since $\mathcal{E}_1$ implies that $a\in J^{(t)}$ and $\varphi(b_{n^\prime})=a$, we can lower bound $\Delta_a$ as:
\begin{align*}
    \expected[\Delta_a]\ge \Pr[\mathcal{E}_1]\cdot\text{deg}(b_{n^\prime})\ge 1/n^\prime\ge c\cdot\text{deg}(a)
\end{align*}
where we have used (\ref{eq:deg-bound}) and the fact that $n^\prime=\text{deg}(a)<2/\eta$.

Now, consider the case where $n^\prime\ge 2/\eta$. Partition the neighbors of $a$ into high and low degree elements: let $p=\lfloor (1-\eta)\text{deg}(a)\rfloor$ and define $L(a)=\{b_1,\ldots,b_p\}$ to be the low-degree elements and $H(a)=\{b_{p+1},\ldots,b_{n^\prime}\}$ to be the high-degree elements. Let $\text{select}_a^{(t)}=\{b\in B^{(t)}:\varphi^{(t)}(b)=a\}$ be the set of elements in $B^{(t)}$ for which $\varphi^{(t)}$ select $a$ and $\mathcal{E}_2$ be the event that $|L(a)\backslash \text{select}_a^{(t)}|\le 2\eta|L(a)|$. We will use the following claim\footnote{Observe that the claim isn't affected by the noise we add for differential privacy, since the probability are taken over the randomness in $\varphi$. As a result, we can use the results directly. Even so, we reproduce the proof in the Appendix for completeness.} from \cite{BPT11} to complete the proof:
\begin{claim}[Claim 3.6 in \cite{BPT11}]
    The following are true:
    \begin{enumerate}[label=\roman*.]
        \item For $\gamma \le\eta/\mathrm{deg}(b_p)$, we have $\Pr[\mathcal{E}_2|x_a=1-\gamma]\ge1/2$.
        \item For $b\in H(a)$ and $\gamma\le\eta/\mathrm{deg}(b)$, we have $\Pr[\varphi^{(t)}(b)=a|\mathcal{E}_2,x_a=1-\gamma]\ge1-\eta$.
    \end{enumerate}
\end{claim}
Observe that the event $\mathcal{E}_2$ implies that $|\text{select}_a^{(t)}|\ge n'-\eta n'-2\eta n'=(1-3\eta)n'$. Since $a\in A^{(t)}$, we know from Line 10 of the algorithm that we have
$$n'+\text{Lap}(1/\epsilon_0)\ge (1-\eta)\tilde{D}(a)+\frac{6\log{n}}{(1-3\eta)\epsilon_0}.$$
Since we are conditioning on the fact that $|\text{Lap}(1/\epsilon_0)|\le 2\log{n}/\epsilon_0$, this implies that 
$$n'\ge (1-\eta)\tilde{D}(a)+\frac{4\log{n}}{(1-3\eta)\epsilon_0}.$$ 
In particular, this means that the event $\mathcal{E}_2$ implies that $|\text{select}_a^{(t)}|\ge (1-4\eta)\tilde{D}(a)+4\log{n}/\epsilon_0$.
But by definition in Line 7, this means that the event $\mathcal{E}_2$ implies that $a\in J^{(t)}$ (again since we are conditioning on the fact that $|\text{Lap}(1/\epsilon_0)|\le 2\log{n}/\epsilon_0$). As a result, we can apply the claim above to obtain
\begin{align*}
    \mathbb{E}[\Delta_a]&\ge\textstyle\sum_{b\in H(a)}\text{deg}(b)\Pr[\mathcal{E}_2\wedge\varphi^{(t)}(b)=a]\\
    &\ge \textstyle\sum_{b\in H(a)}\int_{\gamma=0}^{\eta/\text{deg}(b)}\text{deg}(b)\Pr[\mathcal{E}_2|x_a=1-\gamma]\Pr[\varphi^{(t)}(b)=a|\mathcal{E}_2,x_a=1-\gamma]\,d\gamma\\
    &\ge \textstyle\sum_{b\in H(a)}\eta\frac{1}{2}(1-\eta)\ge c\cdot\text{deg}(a),
\end{align*}
where the final inequality follows since $|H(a)|\ge\eta n^\prime\ge 1$.
\end{proof}

\begin{theorem}
\label{thm:approximate-manis}
For any constant $\eta>0$, we can choose constants $C_1,C_2$ so that Algorithm \ref{alg:manis} satisfies:
\begin{itemize}
    \item Algorithm \ref{alg:manis} outputs a  $(O(\log^2{n}/\epsilon),O(\log^2{n}/\epsilon))$-approximate $\eta$-MaNIS with probability $1-\tilde{O}(1/n^2)$.
    \item Algorithm \ref{alg:manis} is $\epsilon$-edge differentially private.
\end{itemize}
\end{theorem}
\begin{proof}
In order to show that it outputs an approximate MaNIS, we will show that that $A^{(t+1)}=\emptyset$ with high probability at the end of the algorithm. Let $\mathcal{E}$ denote the event that $|\text{Lap}(\beta)|\le 2\beta\log{n}$ for each instance of a $\text{Lap}(\beta)$ random variable. Let us condition on event $\mathcal{E}$, which occurs with probability at least $1-\tilde{O}(1/n^2)$ by properties of the Laplace distribution and a union bound over the $O(\log{n})$ occurrences of the Laplace distribution. Importantly, we will only condition on $\mathcal{E}$ for the utility proof, and not for the privacy proof.

By Lemma \ref{lem:potential function}, we have that $\mathbb{E}[\Phi(t+1)]\le (1-c)\cdot \mathbb{E}[\Phi(t)]$ for each iteration $t$. This implies that after $C_2\log{n}\coloneqq -10\log_{1-c}(n)$ iterations, we have $\expected[\Phi(t)]\le 1/n^3$ so $\Phi(t)=0$ with probability at least $1-O(1/n^2)$; let us also condition on the event that $\Phi(t)=0$ for the remainder of the proof. This implies that $A^{(t+1)}=\emptyset$ since in the condition for $a\in A^{(t+1)}$ in Line 10, the left-hand side is less than $\frac{2\log{n}}{\epsilon_0}$ and the right hand side is greater than $\frac{2\log{n}}{\epsilon_0}$ since we are conditioning on $\mathcal{E}$. Thus, no elements $a\in A^{(t)}\backslash J^{(t)}$ are included in $A^{(t+1)}$ so $A^{(t+1)}=\emptyset$ at the end of the algorithm.

Now, we will show that the algorithm outputs an approximate MaNIS. For the maximality condition, observe that any $a\in A\backslash J$ was not included in $A^{(t+1)}$ in Line 10. This implies that $$|N_{G^{(t)}}(a)\cap B^{(t+1)}|<(1-\eta)\cdot \tilde{D}(a)+\frac{6\log{n}}{(1-3\eta)\epsilon_0}+\text{Lap}(1/\epsilon_0)\le (1-\eta)\cdot N(a)+\frac{12\log{n}}{(1-3\eta)\epsilon_0},$$
since we have conditioned on $\mathcal{E}$. But since $N(J)$ is a superset of the complement of $B^{(t+1)}$, we have
$$N(a)\backslash N(J)\subseteq N(a)\cap B^{(t+1)}=N_{G^{(t)}}(a)\cap B^{(t+1)}.$$
Combined with the previous inequality, this implies that
$$|N(a)\backslash N(J)|\le (1-\eta)\cdot N(a)+\frac{12\log{n}}{(1-3\eta)\epsilon_0},$$
giving our desired maximality condition. 

Next, we prove the independence condition. Fix an index $i\in[k]$. Since $s_i$ was included in the $J$, it was chosen in some $J^{(t)}$. By definition in Line 7, this implies that
$$\sum_{b\in B^{(t)}}\mathbbm{1}\{\varphi^{(t)}(b)=s_i\}\ge (1-4\eta)\cdot\tilde{D}(s_i)-\text{Lap}(1/\epsilon_0)\ge (1-4\eta)\cdot D(s_i)-\frac{6\log{n}}{\epsilon_0}.$$
Now, let $b\in B^{(t)}$ satisfy $\varphi^{(t)}(b)=s_i$. We have that $b$ is not covered by $s_j$'s chosen in previous iterations $\tau<t$, since $b\in B^{(t)}$. We furthermore have that $b$ is not covered by $s_j$'s from the current iteration $t$ with smaller ``$x_a$ value'' than $s_i$, since in that case $\varphi^{(t)}(b)=s_i$ would not be true. Hence we have that $b\not\in N(\{s_1,\ldots,s_{i-1}\})$, giving the inequality
$$|N(s_i)\backslash N(\{s_1,\ldots,s_{i-1}\})|\ge \sum_{b\in B^{(t)}}\mathbbm{1}\{\varphi^{(t)}(b)=s_i\}.$$ 
Combining the two inequalities gives us our desired independence condition.

Finally, we will prove the privacy guarantee. Observe that the private information is only used in Lines 3, 7, and 10 via the Laplace mechanism. It is easy to verify that the sensitivity of the queries is 2 in Line 3 and 1 in Lines 7 and 10, so each instance of the Laplace mechanism is $\epsilon_0$-differentially private (Lemma \ref{lem:laplace mechanism}). Consequently, the entire algorithm is $\epsilon_0\cdot (2C_2\cdot \log{n} +1)$-differentially private by adaptive basic composition (Lemma \ref{lem:basic composition}). Given the chosen constant $C_2$ in the previous paragraphs, we can set $C_1$ to be sufficiently large so that the entire algorithm satisfies $\epsilon$-differential privacy, as desired.
\end{proof}

\section{Differentially Private Max Cover and Set Cover}

The classical greedy algorithm for Max Coverage iteratively selects the set which covers the most additional elements. Since this approach is inherently sequential, it is difficult to adapt it to our setting while preserving pure differential privacy. We instead consider a parallel algorithm which maintains an approximate greedy guarantee while using only $\text{polylog}(n)$ parallel rounds to select the sets~\cite{BPT11}. Similar to the observations of~\cite{DBLP:conf/focs/DhulipalaLRSSY22,dhulipala2024nearoptimal}, we find that this non-private algorithm
being naturally parallel makes it more naturally amenable to modifications which guarantee differential privacy, as we will now show.

The algorithm proceeds in $O(\log{n})$ iterations. In each iteration, the algorithm identifies the collection of sets whose number of additional elements covered in the universe is within an $1+\eta$-multiplicative error of optimal (with $\text{polylog}(n)$ additive error due to privacy). Within this collection, the algorithm uses DP-MaNIS to choose sets to include in the solution. Intuitively, this makes sure that the chosen sets don't have too much overlap so that they still provide near-optimal coverage. The algorithm then outputs the concatenation of the sets which are chosen by DP-MaNIS at each iteration as the final solution. 

In fact, our algorithm will output an ordering of the sets in $\mathcal{S}$, using the ordered output of DP-MaNIS. For Max Cover, the solution is the first $k$ sets in the ordering. For Set Cover, we may use the same algorithm and output the full ordering as the implicit set cover solution. In the remainder of the section, we will show that the ordering obtains good utility guarantees for Max Cover and Set Cover.

\begin{algorithm}[h]
\caption{Differentially Private Max Cover and Set Cover}
\label{alg:k-set-cover}
\textbf{Input:} Universe $\mathcal{U}$, set system $\mathcal{S}$, privacy parameter $\epsilon>0$, and multiplicative error $\eta>0$.\\
\begin{algorithmic}[1]
\STATE $T\leftarrow \log_{1/(1-\eta)}(n), \mathcal{U}_T\leftarrow\mathcal{U}, \mathcal{S}_T\leftarrow\mathcal{S}, \epsilon_0\leftarrow \epsilon/2T$.
\FOR{$t=T,\ldots,1$}
\STATE Define the bucket of sets: $\mathcal{A}_t\leftarrow\{S\in\mathcal{S}_t: 1/(1-\eta)^{t-1}-C_1\log^2(n)/\epsilon_0\le |S\cap\mathcal{U}_t|+\text{Lap}(1/\epsilon_0)\le 1/(1-\eta)^t+C_1\log^2(n)/\epsilon_0\}$
\STATE Select a maximal nearly independent set from the bucket: $J_t\leftarrow\text{DP-MaNIS}_\eta(\mathcal{U}_t,\mathcal{A}_t)$
\STATE Remove elements covered by $J_t$ from the universe: $\mathcal{U}_{t-1}\leftarrow\mathcal{U}_t\backslash \bigcup_{S\in J_t}S$
\STATE Remove the chosen sets $J_t$ from consideration: $\mathcal{S}_{t-1}\leftarrow\mathcal{S}_t\backslash J_t$
\ENDFOR
\STATE Output the concatenation of $J_T,\ldots, J_1, \mathcal{S}_0$
\end{algorithmic}
\end{algorithm}

\subsection{Approximate Prefix Optimality}

We first need to introduce the notion of prefix optimality. In the standard greedy algorithm, the $i^{th}$ set chosen maximizes the number of additional elements covered. For the parallel version in~\cite{BPT11} which our algorithm is based on, they show their algorithm $\mathcal{A}$ satisfies a weaker property called prefix optimality (first introduced in~\cite{chierichetti2010max}) which is sufficient to prove an approximation guarantee for Max Cover and Set Cover. Suppose we fix a choice of the first $i-1$ sets. Given the respective choices of the $i^{th}$ set, let $\alpha_i$ denote the ratio of the number of new elements covered by $\mathcal{A}$ to the number covered by greedy algorithm. Letting $\overline{\alpha}_i=(1/i)\sum_{j\le i}\alpha_j$ be the average of $\alpha_1,\ldots,\alpha_i$, they showed that $\mathcal{A}$ is an $1-\exp(-\overline{\alpha}_k)$ approximation algorithm. Due to the additive error required in order to guarantee differential privacy, we further relax the prefix optimality condition to allow an additive error. 

\begin{definition} 
Let $i\in[n]$ be arbitrary and fix a choice of the first $i-1$ sets $S^{(1)},\ldots,S^{(i-1)}$. Let $S\in\mathcal{S}$ be the $i^{th}$ set chosen by our algorithm $\mathcal{A}$ and let $T\in\mathcal{S}$ be the set which the greedy algorithm would choose. We say that $\mathcal{A}$ is $(\alpha,\beta)$-prefix optimal if for some $\{\alpha_i\}_{i=1}^{n}$, we have $\overline{\alpha}_i\ge \alpha$ for all $i\in[n]$ and $|S-\bigcup_{j=1}^{i-1}S^{(i)}|\ge\alpha_i\cdot|T-\bigcup_{j=1}^{i-1}S^{(i)})|-\beta$ for every $i\in[n]$. 
\end{definition}

We now show that this additive relaxation of prefix optimality implies the same multiplicative approximation guarantee for the Max Cover problem with additional additive error.

\begin{lemma}
If $\mathcal{A}$ is $(\alpha,\beta)$-prefix optimal, then $\mathcal{A}$ is a $(1-\exp(-\bar{\alpha}_k),\beta k)$-approximation for Max Cover.\label{lem:prefix-optimal-max-cover}
\end{lemma}
\begin{proof}
Let $\mathcal{S}=\{S^{(1)},\ldots,S^{(k)}\}$ be the first $k$ sets chosen by $\mathcal{A}$. Let $\mathcal{S}_*=\{S^{(1)}_*,\ldots,S^{(k)}_*\}$ be $k$ sets in an optimal solution (sorted in an arbitrary order). Let $\omega_i$ be the number of elements covered by set $S_i$ which were not covered by $S_1,\ldots,S_{i-1}$. Then the coverage of our algorithm is $\text{cov}(\mathcal{S})=\sum_{i=1}^kS^{(i)}\cdot\omega_i$, where $$\text{cov}(\mathcal{S})\coloneqq |\textstyle\bigcup_{i=1}^{k}S^{(i)}|.$$ 
For any $i\in[k]$, the number of elements covered by $\mathcal{S}_*$ but not covered by $S^{(1)},\ldots,S^{(i-1)}$ is at least $$\text{cov}(\mathcal{S}_*)-\textstyle\sum_{j=1}^{i-1}\omega_i.$$ This implies that there exists some $S\in\mathcal{S}_*$ not already chosen in $S^{(1)},\ldots,S^{(i-1)}$ which covers at least $$[{\text{cov}(\mathcal{S}_*)-\textstyle\sum_{j=1}^{i-1}\omega_i}]/k$$ elements. Since our algorithm chooses an $(\alpha_i,\beta)$-approximately greedy solution, we have $$\omega_i\ge\alpha_i\cdot [{\text{cov}(\mathcal{S}_*)-\textstyle\sum_{j=1}^{i-1}\omega_i}]/k-\beta.$$ We will now prove by induction that $$\textstyle\sum_{j=1}^{i}\omega_j\ge \left(1-\left(1-{\overline{\alpha}_i}/{k}\right)^i\right)\cdot\text{cov}(\mathcal{S}_*)-i\beta$$ for each $i\in[k]$. This directly implies our desired result since $(1-x/k)^k\le\exp(-x)$. The base case $i=1$ is obvious since $$\omega_1\ge\alpha_1\cdot\frac{\text{cov}(\mathcal{S}_*)}{k}-\beta=\overline{\alpha}_1\cdot\frac{\text{cov}(\mathcal{S}^*)}{k}-\beta.$$
Now suppose the claim holds for $i$. We will prove it for $i+1$:
\begin{align}
    \textstyle\sum_{j=1}^{i+1}\omega_j&=\omega_{i+1}+\textstyle\sum_{j=1}^{i}\omega_j\nonumber\\
    &\ge \alpha_{i+1}\cdot\frac{\text{cov}(\mathcal{S}_*)-\textstyle\sum_{j=1}^{i}\omega_j}{k}+\textstyle\sum_{j=1}^{i}\omega_j-\beta\label{eq:1}\\
    &= \alpha_{i+1}\cdot\frac{\text{cov}(\mathcal{S}_*)}{k}+ \left(1-\frac{\alpha_{i+1}}{k}\right)\cdot\textstyle\sum_{j=1}^{i}w_j-\beta\nonumber\\
    &\ge \alpha_{i+1}\cdot\frac{\text{cov}(\mathcal{S}_*)}{k}+ \left(1-\frac{\alpha_{i+1}}{k}\right)\cdot\left(1-\left(1-\frac{\overline{\alpha}_{i+1}}{k}\right)^{i+1}\right)\text{cov}(\mathcal{S}_*)-(i+1)\beta\label{eq:2}\\
    &\ge \left(1-\left(1-\frac{\overline{\alpha}_{i+1}}{k}\right)^{i+1}\right)\cdot\text{cov}(\mathcal{S}_*)-(i+1)\beta,\label{eq:3}
\end{align}
where inequality (\ref{eq:1}) follows by our bound on $\omega_{i+1}$, inequality (\ref{eq:2}) follows by the inductive hypothesis, and inequality (\ref{eq:3}) follows by the inequality $\left(1-\frac{\overline{\alpha}_i}{k}\right)^i\cdot\left(1-\frac{\alpha_{i+1}}{k}\right)\le\left(1-\frac{\overline{\alpha}_{i+1}}{k}\right)^{i+1}$ (see Lemma 3 in \cite{chierichetti2010max}) and rearranging terms.
\end{proof}

Next, we show that prefix optimality also implies a multiplicative approximation algorithm for Set Cover. 

\begin{lemma}
    If $\mathcal{A}$ is $(1-\eta,\beta)$-prefix optimal, then $\mathcal{A}$ is a $O(\beta+\log{n})$-approximation for implicit Set Cover.\label{lem:prefix-optimal-set-cover}
\end{lemma}
\begin{proof}
    As before, let $S^{(1)},\ldots,S^{(m)}$ denote the ordering of sets chosen by the algorithm $\mathcal{A}$. For each index $i\in[m]$, let $R_i=\{u\in\mathcal{U}:u\not\in\bigcup_{j=1}^{i-1}S^{(i)}$ denote the elements which remain uncovered and let $L_i=\max_{S\in\mathcal{S}}|S\cap R_i|$ denote the maximum number of elements a new set can cover. By definition of prefix optimality, we know that the $i^{th}$ set which is chosen by $\mathcal{A}$ covers at least $(1-\eta)\cdot L_i-\beta$ additional elements. For iterations where $L_i>2\beta/(1-\eta)$, this implies that $\mathcal{A}$ covers $(1-\eta)\cdot L_i/2$ elements. By the classical set cover argument, this implies that at most $O(\text{OPT}\cdot\log(n))$ sets are chosen in these iterations. When $L_i<2\beta/(1-\eta)$, observe that the number of remaining elements $|R_i|$ is at most $\text{OPT}\cdot L_i$. Any permutation will incur cost at most $\text{OPT}\cdot L_i=O(\text{OPT}\cdot\beta)$ for these elements. Combining our two bounds gives the desired result of a $O(\beta+\log(n))$-approximation algorithm for implicit Set Cover.
\end{proof}

\subsection{Privacy and Utility Guarantees}

Finally, we show the privacy and utility guarantees. For the privacy guarantees, we simply show that the algorithm is an (adaptive) composition of Laplace mechanisms, so it immediately satisfies differential privacy. For utility, we show that the algorithm is $(\alpha,\beta)$-prefix optimal for $\alpha=1-O(\eta)$ and $\beta=O(\log^3(n)/\epsilon)$.

\begin{theorem}
    Algorithm \ref{alg:k-set-cover} is $\epsilon$-edge differentially private.
\end{theorem}
\begin{proof}
    The private information is only used in Line 3 through the Laplace mechanism and in Line 4 through a call to our algorithm, DP-MaNIS. It is easy to verify that the sensitivity of the vector consisting of $|S\cap\mathcal{U}_t|$ for each $S\in\mathcal{S}_t$ is $1$, so Line 3 is $\epsilon_0$-edge differentially private by the privacy of the Laplace mechanism (Lemma \ref{lem:laplace mechanism}). Line 4 is also $\epsilon_0$-edge differentially private by the second bullet in Theorem \ref{thm:approximate-manis}. Thus, the entire algorithm is $2T\epsilon_0$-edge differentially private by adaptive basic composition (Lemma \ref{lem:basic composition}). By our choice of $\epsilon_0$ and $T$, this implies our desired guarantee.
\end{proof}

\begin{theorem}
    Algorithm \ref{alg:k-set-cover} is $(1-5\eta,O(\log^2{n}/\epsilon_0))$-prefix optimal with probability $1-\tilde{O}(1/n^2)$.\label{thm:prefix-optimal}
\end{theorem}
\begin{proof}
    For the entirety of the proof, we will condition on the event $\mathcal{E}$ that all $\text{Lap}(\beta)$ random variables satisfy $|\text{Lap}(\beta)|\le 2\beta\log{n}$. As before, this event occurs with probability at least $1-\tilde{O}(1/n^2)$. 

    Fix an index $i\in[n]$ in the output, let $s_i$ denote the $i^{th}$ set of the output, and let $t$ denote the iteration when the set was chosen. Also let $s^*$ denote the set which greedy would choose. We wish to show that 
    \begin{align}
        \left|s_i\backslash \textstyle\bigcup_{j=1}^{i-1}s_j\right|\ge (1-5\eta)\cdot\left|s^*\backslash\textstyle\bigcup_{j=1}^{i-1}s_j\right|-O(\log^3{n}/\epsilon).\label{eq:to-show}
    \end{align}
    We claim that at iteration $t$, all remaining sets $s\in\mathcal{S}_t$ have utility 
    \begin{align}
        |s\cap\mathcal{U}_t|\le 1/(1-\eta)^t+O(\log^2(n)/\epsilon_0).  \label{eq:utility-bound}  
    \end{align}
    We prove this claim inductively on $t$, going downwards. At iteration $T$, this is trivially true since all sets have size at most $n$. Assume the claim holds at iteration $t+1$. Recall that the set $\mathcal{S}_t$ consists of those sets $S\in \mathcal{S}_{t+1}$ which were not chosen in $\mathcal{A}_{t+1}$, or were chosen in $\mathcal{A}_{t+1}$ but not chosen in $J_{t+1}$. We know that all sets $s\in\mathcal{S}_{t+1}$ satisfying $$|s\cap\mathcal{U}_{t+1}|\ge 1/(1-\eta)^t-C_1\log^2(n)/\epsilon_0$$ are included in $\mathcal{A}_{t+1}$. This implies that if set $s\in\mathcal{S}_{t+1}$ isn't chosen in $\mathcal{A}_{t+1}$, it must satisfy Inequality \ref{eq:utility-bound}. 
    
    Now, we consider the sets $s$ chosen for $\mathcal{A}_{t+1}$ but not chosen for $J_{t+1}$. Since we apply DP-MaNIS on $\mathcal{A}_{t+1}$ to obtain $J_{t+1}$, we know that each set in $\mathcal{A}_{t+1}$ which isn't included in $J_{t+1}$ has decreased their utility by a factor of $1-\eta$. Specifically, the second part of Definition \ref{def:MaNIS} guarantees that $s\in\mathcal{A}_{t+1}\backslash J_{t+1}$ satisfies
    $$|s\cap\mathcal{U}_{t}|\le (1-\eta)\cdot|s\cap\mathcal{U}_{t+1}|+O(\log^2(n)/\epsilon_0)\le1/(1-\eta)^t+O(\log^2(n)/\epsilon_0),$$
    where the second inequality follows by the inductive hypothesis. That completes the inductive proof.

    Next, we need to obtain a lower bound on the utility of the chosen $s_i$. We claim that $s_i$ satisfies 
    \begin{align}
        \left|s_i\backslash\textstyle\bigcup_{j=1}^{i-1}s_j\right|\ge (1-4\eta)\cdot 1/(1-\eta)^{t-1}-O(\log^2(n)/\epsilon_0).\label{eq:utility-lower-bound}
    \end{align}
    Observe that since $s_i$ was chosen in iteration $t$, we have that $s_i\in\mathcal{A}_t$ implying that $$|s_i\cap\mathcal{U}_t|\ge 1/(1-\eta)^{t-1}-O(\log^2(n)/\epsilon_0).$$
    Furthermore, since $s_i$ was chosen using DP-MaNIS, we know that the chosen sets are nearly independent so $|s_i\backslash \bigcup_{j=1}^{i-1}s_j|$ and $|s_i\cap\mathcal{U}_t|$ are similar size. Specifically, the first part of Definition \ref{def:MaNIS} guarantees that 
    $$\left|s_i\backslash \textstyle\bigcup_{j=1}^{i-1}s_j\right|\ge (1-4\eta)\cdot |s_i\cap\mathcal{U}_t|-O(\log^2(n)/\epsilon_0).$$
    Combining the above two inequalities gives us inequality (\ref{eq:utility-lower-bound}), as desired.
    
    Finally, we combine the inequalities (\ref{eq:utility-bound}) and (\ref{eq:utility-lower-bound}) to conclude our desired result of (\ref{eq:to-show}). In particular, inequality (\ref{eq:utility-bound}) applies to all $s\in\mathcal{S}_t$ so it also applies to the set $s^*$ chosen by greedy. Also observe that that $\mathcal{U}\backslash\bigcup_{j=1}^{i-1}s_j\supseteq\mathcal{U}_t$ since $\mathcal{U}_t$ is the set of elements remaining in $\mathcal{U}$ before the iteration starts. Combining the two inequalities and then applying our observation gives us
    $$\left|s_i\backslash \textstyle\bigcup_{j=1}^{i-1}s_j\right|\ge (1-5\eta)\cdot|s^*\cap\mathcal{U}_t|-O(\log^2(n)/\epsilon_0)\ge (1-5\eta)\cdot \left|s^*\backslash\textstyle\bigcup_{j=1}^{i-1}s_j\right|,$$
    proving the desired claim. 
\end{proof}

\noindent
Finally, the proofs of Theorems \ref{thm:intro-max-cover} and \ref{thm:intro-set-cover} follow by combining Theorem \ref{thm:prefix-optimal} with Lemmas \ref{lem:prefix-optimal-max-cover} and \ref{lem:prefix-optimal-set-cover}. 

\section{Lower Bounds}

We give a lower bound showing that the additive error given in our algorithm is near-optimal for $\epsilon$-edge differentially private algorithms. The lower bound is inspired by the one given in~\cite{GuptaLMRT10} for max coverage under node differential privacy.

\begin{theorem}
    Any $\epsilon$-edge differentially private algorithm for max coverage incurs $\Omega(k\log(n/k)/\epsilon)$ additive error. 
\end{theorem}
\begin{proof}
    Observe that the lower bound of $\Omega(k\log(n/k)/\epsilon)$ for node-private Max Cover (Theorem 4.5) in~\cite{GuptaLMRT10} is for instances where each element is only in a single set. For such instances, any $\epsilon$-edge differentially private mechanism can be converted to an $\epsilon$-node differentially private mechanism (see Lemma \ref{lem:equivalent-f}). Hence, the lower bound of $\Omega(k\log(n/k)/\epsilon)$ also applies to edge-private Max Cover.
\end{proof}

\bibliographystyle{alpha}
\bibliography{refs.bib}

\end{document}